\RequirePackage{fix-cm}
\documentclass[smallextended]{svjour3}       
\smartqed  
\usepackage{graphicx}
\begin{document}

\title{An $O(n^2)$ algorithm for Many-To-Many Matching of Points with Demands in One Dimension
}


\author{Fatemeh Rajabi-Alni         \and
        Alireza Bagheri 
}


\institute{F. Rajabi-Alni \at
              Department of Computer Engineering and IT,\\Amirkabir University of Technology, Tehran, Iran.
              \email{f.rajabialni@aut.ac.ir \\ (fatemehrajabialni@yahoo.com)}
           \and
           A. Bagheri \at
              Department of Computer Engineering and IT,\\Amirkabir University of Technology, Tehran, Iran.
}

\date{Received: date / Accepted: date}

\maketitle

\begin{abstract}

Given two point sets $S$ and $T$, we study the many-to-many matching with demands problem (MMD problem) which is a generalization of the many-to-many matching problem (MM problem). In an MMD, each point of one set must be matched to a given number of the points of the other set (each point has a demand). In this paper, we consider a special case of MMD problem, the one-dimensional MMD (OMMD), where the input point sets $S$ and $T$ lie on the line. In OMMD problem, the cost of matching a pair of points is equal to the distance between the two points. We present the first $O\left(n^2\right)$ time algorithm for computing an OMMD between $S$ and $T$, where $\left|S\right|+\left|T\right|=n$.
\end{abstract}

\section{Introduction}
\label{intro}

Suppose we are given two point sets $S$ and $T$, a \textit {many-to-many matching} (MM) between $S$ and $T$ assigns each point of one set to one or more points of the other set \cite{ColanDamian}. Eiter and Mannila \cite{Eiter} solved the MM problem using the Hungarian method in $O(n^3)$ time. Finally, Colannino et al. \cite{ColanDamian} presented an $O(n \log {n})$-time dynamic programming solution for finding an MM between two sets on the real line. The matching has different applications such as computational biology \cite{Ben-Dor}, operations research \cite{Burkard}, pattern recognition \cite{Buss}, and computer vision \cite{Fatih}.

A general case of MM problem is \textit {the limited capacity many-to-many matching problem} (LCMM) where each point has a capacity. Schrijver \cite{Schrijver} proved that a minimum-cost LCMM can be found in strongly polynomial time. A special case of the LCMM problem is that in which both $S$ and $T$ lie on the real line. Rajabi-Alni and Bagheri \cite{Rajabi-Alni} proposed an $O(n^2)$ time algorithm for the one dimensional minimum-cost LCMM.

In this paper we consider another generalization of the MM problem, where each point has a \textit{demand}, that is each point of one set must be matched to a given number of the other set. Let $S=\{s_1,s_2,\dots,s_y\}$ and $T=\{t_1,t_2,\dots,t_z\}$. We denote the demand sets of $S$ and $T$ by $D_S=\{\alpha_1,\alpha_2,\dots,\alpha_y\}$ and $D_T=\{\beta_1,\beta_2,\dots,\beta_z\}$, respectively. In a many-to-many matching with demand (MMD), each point $s_i \in S$ must be matched to $\alpha_i$ points in $T$ and each point $t_i \in T$ must be matched to $\beta_i$ points in $S$. We denote the demand of each point $a \in S \cup T$ by $demand(a)$. We study one dimensional MMD (OMMD), where $S$ and $T$ lie on the line and propose an $O(n^2)$ algorithm for finding a minimum cost OMMD.

\section{Preliminaries}
\label{PreliminSect}
In this section, we proceed with some useful definitions and assumptions. Fig. \ref{fig:1} provides an illustration of them. Let $S=\{s_i \  for \  1 \le i \le y\}$ and $T=\{t_i \  for  \ 1 \le i \le z\}$. We denote the elements in $S$ in increasing order by
$(s_1, . . . , s_y )$, and the elements in $T$ in increasing order by $(t_1, . . . , t_z )$. Let $s_1$ be the smallest point in $S \cup T$. Let $S \cup T$ be partitioned into maximal subsets $A_0,A_1,A_2,\dots $ alternating between subsets in $S$ and $T$ such that all points in $A_i$ are smaller than all points in $A_{i+1}$ for all $i$: the point of highest coordinate in $A_i$ lies to the left of the point of lowest coordinate in $A_{i+1}$ (Fig. \ref{fig:1}).

Let $A_w=\{a_1,a_2,\dots,a_s\}$ with $a_1< a_2<\dots<a_s$ and $A_{w+1}=\{b_1,b_2,\dots,b_t\}$ with $b_1< b_2<\dots<b_t$. We denote $|b_1-a_i|$ by $e_i$, $|b_i-b_1|$ by $f_i$. Obviously $f_1=0$. Note that $a_0$ represents the largest point of $A_{w-1}$ for $w>0$. In an OMMD, a point matched to at least its demand number of points is a \textit{satisfied} point.

 \begin{figure}
\vspace{-7cm}
\hspace{-13cm}
\resizebox{3\textwidth}{!}{%
  \includegraphics{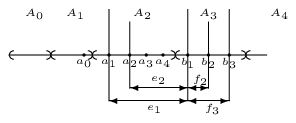}
}
\vspace{-40cm}
\caption{The notation and definitions in partitioned point set $S \cup T$.}
\label{fig:1}       
\end{figure}

First, we briefly describe the algorithm in \cite{ColanDamian}. The running time of their algorithm is $O(n \log n)$ and $O(n)$ for the unsorted and sorted point sets $S$ and $T$, respectively. We denote by $C(q)$ the cost of a minimum MM for the set of the points $\left\{p\in S\cup T\right|p\leq q\}$. The algorithm in \cite{ColanDamian} computes $C(q)$ for all points $q$ in $S\cup T$. Let $m$ be the largest point in $S \cup T$, then the cost of the minimum MM between $S$ and $T$ is equal to $C(m)$.


 \begin{figure}
\vspace{-6cm}
\hspace{-12cm}
\resizebox{3\textwidth}{!}{%
  \includegraphics{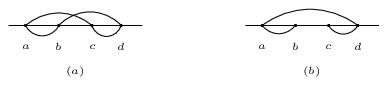}
}
\vspace{-40.5cm}
\caption{Suboptimal matchings. (a) $(a,c)$ and $(b,d)$ do not both belong to an optimal matching. (b) $(a,d)$ does not belong to an optimal matching. }
\label{fig:2}       
\end{figure}

\begin{lemma}
\label{lem3}
Let $b<c$ be two points in $S$, and $a<d$ be two points in $T$ such that $a\leq b<c\leq d$. Then a minimum cost many-to-many matching that contains $(a,c)$ does not contain $(b,d)$, and vice versa (Fig. \ref{fig:2}(a)) \cite{ColanDamian}.
\end{lemma}

\begin{lemma}
\label{lem1}
Let ${b,d} \in T$ and ${a,c} \in S$ with $a<b<c<d$. Then, a minimum cost many-to-many matching contains no pairs $(a,d)$ (Fig. \ref{fig:2}(b)) \cite{ColanDamian}.
\end{lemma}

Towards a contradiction, suppose that $M$ is a minimum cost MM that contains such a pair $(a,d)$ (Fig. \ref{fig:2}(b)). We can construct a new MM, denoted by $M'$, by removing the pair $(a,d)$ from $M$ and adding the pairs $(a,b)$ and $(c,d)$. The cost of $M'$ is smaller than the cost of $M$. This is a contradiction to our assumptions that $M$ is a mnimum-cost MM.

\begin{corollary}
\label{consecutive}
Let $M$ be a minimum-cost MM. For any matching $(a,d) \in M$ with $a < d$, we have $a \in A_i$ and $d \in A_{i+1}$ for some $i \geq 0$.
\end{corollary}

\begin{lemma}
\label{lem4}
In a minimum cost MM, each $A_i $ for all $i>0$ contains a point $q_i$, such that all points $a \in A_i$ with $a<q_i$ are matched with the points in $A_{i-1}$ and all points $a' \in A_i$ with $q_i<a'$ are matched with the points in $A_{i+1}$ \cite{ColanDamian}.
\end{lemma}

\begin{proof}
By Corollary \ref{consecutive}, each pint $a \in A_i$ must be matched with a point $b$ such that either $b \in A_{i-1}$ or $b \in A_{i+1}$. Let $b \in A_{i-1}$, $b' \in A_{i+1}$, and $a,a' \in A_i$ with $b<a<a'<b'$. By way of contradiction, suppose that $M$ is a minimum cost MM containing both $(b,a')$ and $(a,b')$. Contradiction with Lemma \ref{lem3}.

\end{proof}

The point $q_i$ defined in Lemma \ref {lem4} is called \textit{the separating point}. In fact, the aim of their algorithm is exploring the separating points of each partition $A_i$ for all $i$. They assumed that $C(p)=\infty$ for all points $p \in A_0$. Let $A_w=\{a_1,a_2,\dots,a_s\}$ and $A_{w+1}=\{b_1,b_2,\dots,b_t\}$. Their dynamic programming algorithm computes $C(b_i)$ for each $b_i \in A_{w+1}$, assuming that $C(p)$ has been computed for all points $p < b_i$ in $S\cup T$. Depending on the values of $w$, $s$, and $i$ there are five possible cases.

\begin{description}
\item[Case 0:] $w=0$. In this case there are two possible situations.
\begin{itemize}
\item $i\leq s$. We compute the optimal matching by assigning the first $s-i$ elements of $A_0$ to $b_1$ and the remaining $i$ elements pairwise (Fig. \ref{fig:3}(a)). So we have
\[C\left(b_i\right)=\sum^s_{j=1}{e_j}+\sum^i_{j=1}{f_j}.\]

\item $i>s$. The cost is minimized by matching the first $s$ points in $A_1$ pairwise with the points in $A_0$, and the remaining $i-s$ points in $A_1$ with $a_s$ (Fig. \ref{fig:3}(b)). So we have
\[C\left(b_i\right)=\left(i-s\right)e_s+\sum^s_{j=1}{e_j}+\sum^i_{j=1}{f_j}.\]

\end{itemize}

\begin{figure}
\vspace{-3cm}
\hspace{-6cm}
\resizebox{2\textwidth}{!}{%
  \includegraphics{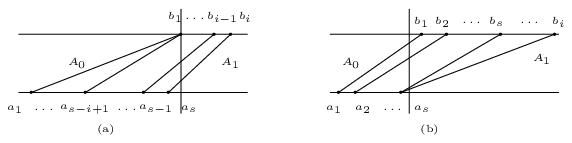}
}
\vspace{-26cm}
\caption{Case 0: $w=0$. (a) $1 \leq i \leq s$. (b) $s<i \leq t$.}
\label{fig:3}       
\end{figure}

\item[Case 1:] $w>0,s=t=1$. Fig. \ref{fig:4}(a) provides an illustration of this case. By Lemma \ref{lem4}, $b_1$ must be matched with the point $a_1$. Therefore, we can omit the point $a_1$, unless it reduces the cost of $C\left(b_1\right)$.
\newline
\newline

\item[Case 2:] $w>0,s=1,t>1$. By Lemma \ref{lem4}, we can minimize the cost of the many-to-many matching by matching all points in $A_{w+1}$ with $a_1$ as presented in Fig. \ref{fig:4}(b). As case 1, $C\left(b_i\right)$ includes $C\left(a_1\right)$ if $a_1$ covers other points in $A_{w-1}$; otherwise, $C\left(b_i\right)$ includes $C\left(a_0\right)$.
\newline
\newline
\item[Case 3:] $w>0,s>1,t=1$. By Lemma \ref{lem4}, we should find the point $a_i\in A_w$ such that all points $a_j \in A_w$ with $a_j<a_i$ are matched to points in $A_{w-1}$ and all points $a_k \in A_w$ with $a_i<a_k$ are matched to points in $A_{w+1}$ (Fig. \ref{fig:4}(c)).
\newline
\newline
\begin{figure}
\vspace{-5cm}
\hspace{-6cm}
\resizebox{2\textwidth}{!}{%
  \includegraphics{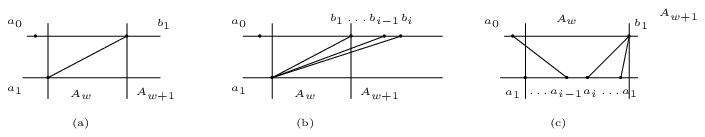}
}
\vspace{-26cm}
\caption{(a). Case 1: $w>0,s=t=1$. (b) Case 2: $w>0,s>1,t=1$. (c) Case 3: $w>0,s>1,t=1$.}
\label{fig:4}       
\end{figure}

\item[Case 4:] $w>0,s>1,t>1$. In this case, we should find the point $q$ that splits $A_w$ to the left and right. Let $X\left(b_i\right)$ be the cost of connecting $b_1,b_2,\dots,b_i$ to at least $i+1$ points in $A_w$ (Fig. \ref{fig:5}(a)). Let $Y\left(b_i\right)$ be the cost of connecting $b_1,b_2,\dots,b_i$ to exactly $i$ points in $A_w$ (Fig. \ref{fig:5}(b)). Finally, let $Z\left(b_i\right)$ denote the cost of connecting $b_1,b_2,\dots,b_i$ to fewer than $i$ points in $A_w$, as depicted in Fig. \ref{fig:5}(c).

Then, we have:
\[C\left(b_i\right)=\left\{
\begin{array}{lr}
\min(X\left(b_i\right),Y\left(b_i\right),Z\left(b_i\right)) & 1\leq i<s
 \\
 \min(Y\left(b_s\right),Z\left(b_s\right)) & i=s
 \\
 C\left(b_{i-1}\right)+e_s+f_i & s<i\leq t
 \end{array}
\right..\]

\end{description}

\begin{figure}
\vspace{-4cm}
\hspace{-9cm}
\resizebox{2.3\textwidth}{!}{%
  \includegraphics{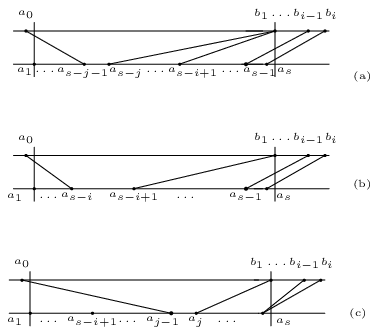}
}
\vspace{-25.5cm}
\caption{Case 4: $w>0,s>1,t>1$. (a) Computing $X(b_i)$. (b) Computing $Y(b_i)$. (c) Computing $Z(b_i)$.}
\label{fig:5}       
\end{figure}

We present another simpler algorithm for Case 4 by Lemma \ref{lemma-new}.

\begin{lemma}
\label{lemma-new}
Give $A_w=\{a_1,a_2,\dots,a_s\}$ with $a_1< a_2<\dots<a_s$ and $A_{w+1}=\{b_1,b_2,\dots,b_t\}$ with $b_1< b_2<\dots<b_t$. In a minimum cost MM, for $i \geq 2$ and $w>0$ we have:

\[C(b_i)=\ \left\{
 \begin{array}{ll}
C(b_{i-1})+f_i-f_1  & if\ deg(b_1)>1
 \\

C(b_{i-1})+e_s+f_i & if\ deg(b_1)=1 \ and \ deg(a_s)>1
 \\
\min(C(b_{i-1})+e_{s}+f_i,C(b_{i-1})+e_{s-i+1}+f_i-C(a_{s-i+1})+C(a_{s-i})), & if\ deg(b_1)=1 \ and \ deg(a_s)=1
 \\

 \end{array}
\right..\]

\end{lemma}

\begin{proof}

We first insert the point $b_1$, and then by Lemma \ref{lem4} we find the separating point $q_i$. Let $q_i=a_h$, that is we get the minimum cost by matching $a_h, \dots, a_s$ to $b_1$.

For the point $b_i$ there are two cases for $2\geq i$:
\begin{itemize}
\item $deg(b_1)>1$. Assume that $b_1$ is matched to $a_j \dots a_s$. An optimal MM can be computed by matching $b_i$ to one of the points $a_j \dots a_s$, say $a_k$. Since if we assume that $b_i$ is matched to $a_s$ in the best case by contradiction. Then, by removing $(a_k,b_1)$ and $(a_s,b_i)$ and adding $(a_k,b_i)$, we get an MM with smaller cost.

\item $deg(b_1)=1 $. In this case, there are two cases:

\begin{itemize}

\item $deg(a_s)>1$. Let $y=deg(a_s)-1$. In this case, $b_{i-y},\dots,b_{i-1}$ are matched to $a_s$. It is proved by induction on $i$ by Lemma \ref{cccc}. for $i=2$ it is obviously true.

\begin{lemma}
\label{cccc}
If $b_{i-1}$ selects $a_s$ as the separating point, then $a_s$ is also the separating point (the optimal point) for $b_{i} \dots b_t$ for $i\ge 3$,.
\end{lemma}

\begin{proof}

If $b_{i-1}$ selects $a_s$ as the separating point, then:

$$C(b_{i-2})+e_s+f_{i-1}<C(b_{i-2})+f_{i-1}+e_k+C(a_{k-1})-C(a_{k}),$$
and so:

$$e_s<e_k+C(a_{k-1})-C(a_{k}),$$
where $a_k$ is the largest point in $A_w$ that is matched to a point of $A_{w-1}$.

If we add $C(b_{i-1})$ and $f_{i}$ we have:

$$C(b_{i-1})+e_s+f_{i}<C(b_{i-1})+f_{i}+e_k+C(a_{k-1})-C(a_{k})$$, so $a_s$ is also the separating point (the optimal point) for $b_{i}$.

\end{proof}
\item $deg(a_s)=1$ and $deg(a_{s-i+1})=1$ and $a_{s-i+1}$ is matched to $A_{w-1}$. Obviously, all $i-1$ points in $b_1,\dots,b_{i-1}$ are matched to $a_s, \dots, a_{s-i+2}$. In this case we have:
$$C(b_{i})=\min(C(b_{i-1})+e_{s}+f_i,C(b_{i-1})+e_{s-i+1}+f_i-C(a_{s-i+1})+C(a_{s-i})).$$ That is, we must select if $b_i$ is matched to $a_s$ or $a_{s-i+1}$.
\item $deg(a_s)=1$ and $deg(a_{s-i+1})>1$ and $a_{s-i+1}$ is matched to $A_{w-1}$. In this case $b_i$ is matched to $a_s$.
\end{itemize}

\end{itemize}

\end {proof}

Similarly, we also present a more simpler algorithm for Case 4 in \cite{Rajabi-Alni} in the following.

For the point $b_i$ there are two cases for $2\geq i$:
\begin{itemize}
\item There exists at least one point $b_k$ in $A_{w+1}$ with $1 \leq k<i$ such that $deg(b_k)>1$, that is $b_k$ is matched to $a_j \dots a_y$ of $A_w$. An optimal MM with limited capacities can be computed by matching $b_i$ to one of the points $a_j \dots a_y$, say $a_q$. So, we examine $b_1$ to $b_{i-1}$ to find the first point $b_k$ that is matched to more than one point. We match $b_i$ to $a_q$, and remove $(b_k,a_q)$ form MM with limited capacities.

\item $deg(b_q)=1$ for all the points $b_q$ in $A_{w+1}$ with $1 \leq q<i$ . In this case, there are two cases:

\begin{itemize}

\item For the largest point in $a_1 , a_2, \dots, a_s$, called $a_q$, that is not matched to $b_i$ we have $1<deg(a_q)<capacity(a_q)$. Since $1<deg(a_q)$, then
    \begin{itemize}
      \item either $a_q$ is the optimal point for more than one point of $A_{w-1}$, and the capacities of the point $a_1, \dots, a_{q-1}$ is saturated by other points from $A_{w-1}$.
      \item Or $a_q$ is the optimal point for more than one point of $A_{w+1}$, especially $b_{i-1}$.
      \item Or both.
      \end{itemize}

 So, $b_i$ is matched to $a_q$. It can be proved by induction on $i$ by Lemma \ref{cc888}. for $i=2$ it is obviously true.

\begin{lemma}
\label{cc888}
If $b_{i-1}$ selects $a_q$ as the separating point, then $a_q$ is also the separating point (the optimal point) for $b_{i}$ for $i\ge 3$,.
\end{lemma}

\item $deg(a_s)=1$ and $capacity(a_s)>1$ and $deg(a_{s-i+1})=1$. Obviously, all $i-1$ points in $b_1,\dots,b_{i-1}$ are matched to $a_s, \dots, a_{s-i+2}$. In this case we examine that if $b_i$ is matched to $a_s$ or $a_{s-i+1}$.

\end{itemize}

\end{itemize}

\section{An Algorithm for OMMD problem}
\label{OMAsection}
In this section, we present an $O(n^2)$ algorithms for finding an OMMD between two sets $S$ and $T$ lying on the line. Our recursive dynamic programming algorithm is based on the algorithm of Colannino et al. \cite{ColanDamian} and Rajabi and Bagheri \cite{Rajabi-Alni}. We begin with some useful lemmas.

\begin{lemma}
\label{lem5}
Let $a\in S,b\in T$ and $c\in S,d\in T$ such that $a\leq b<c\leq d$. Let $M$ be an minimum-cost OMMD. If $(a,d) \in M$, then either $(a,b) \in M$ or $(c,d) \in M$ or both.
\end{lemma}

\begin{proof}
The proof of this lemma is essentially the same as the proof of Lemma \ref{lem1}.
\end{proof}

\begin{corollary}
\label{corlem5}
Let $a \in A_i$ and $d \in A_j$ for some $i \geq 0$. For any matching $(a,d)$ in an OMMD, if $j>i+1$ then either $(a,b) \in M$ for all $b \in A_{i+1} \cup A_{i+3}\cup  \dots \cup A_{j-2}$ or  $(a,b) \in M$ for all $a \in A_{i+2} \cup A_{i+4}\cup  \dots \cup A_{j-1}$.
\end{corollary}

Note that we use this corollary for satisfying the demands of $a \in A_i$ by the points of sets $A_j$ or for satisfying the demands of $b \in A_j$ by the points of sets $A_i$.


\begin{lemma}
\label{lem6}
Let $b<c$ be two points in $S$, and $a<d$ be two points in $T$ such that $a\le b<c\le d$. If an OMMD contains both of $(a,c)$ and $(b,d)$, then $(a,b) \in M$ or $(c,d) \in M$.
\end{lemma}

\begin{proof}
Suppose that the lemma is false. Let $M$ be an OMMD that contains both $(a,c)$ and $(b,d)$, and neither $(a,b) \in M$ nor $(c,d) \in M$ (Fig. \ref{fig:2}(a)). Then we can remove the pairs $(a,c)$ and $(b,d)$ from $M$ and add the pairs $(a,b)$ and $(c,d)$: the result $M'$ is still an OMMD which has a smaller cost, a contradiction.
\end{proof}




\begin{lemma}
\label{lemma-lem10}
Let $a<a'\leq b<b'$ such that $a,a' \in S$ and $b,b' \in T$. Assume that we must match the points $a,a'$ to the points $b,b'$. Then, in an OMMD it does not matter that we use the pairs $(a,b),(a',b')$ or the pairs $(a,b'),(a',b)$.
\end{lemma}
\begin{proof}
The cost of the two pairs $(a,b),(a',b')$ is equal to the cost of the two pairs $(a,b'),(a',b)$. Since we have $$(b-a)+(b'-a')=(b'-a)+(b-a').$$
\end{proof}
\begin{lemma}
\label{lem10}
Let $A=\{a_1,a_2,\dots,a_s\}$ and $B=\{b_1,b_2,\dots,b_t\}$ be two distinct sorted sets of points with demands lying on the real line. Let $D_A=\{\alpha_1,\alpha_2,\dots,\alpha_s\}$ and $D_B=\{\beta_1,\beta_2,\dots,\beta_t\}$ be the demands of the points of $A$ and $B$, respectively, such that $s\ge \max^t_{j=1}{\beta_j}$ and $t\ge \max^s_{i=1}{\alpha_i}$. Then, we can compute an OMMD between $A$ and $B$ in $O(s+t)$ time.
\end{lemma}

\begin{proof}
Obviously, for reaching the minimum cost we must match each point in $B$ with the first unsatisfied point in $A$ (if exists). Notice by Lemma \ref{lemma-lem10}, the order of matching is arbitrary. Without loss of generality we assume $a_s \le b_1$ and match two sets as follows.

First, we match each point $a_i \in A$ with the first unsatisfied point in $B$. If all points in $B$ are satisfied, then we match $a_i \in A$ with the closest point of $B$ that is not matched with it. Then, if there exists any unsatisfied point in $B$, starting from $b_{1}$, we match each point of $B$ with the closest point of $A$ not matched with it previously.

\end{proof}

\begin{theorem}
Let $S$ and $T$ be two sets of points on the real line with $\left|S\right|+\left|T\right|=n.$ Then, an OMMD between $S$ and $T$ can be determined in $O(n^2)$ time.
\end{theorem}

Let $Demand(q)$ denote the demand of the point $q$, i.e., the number of the points that must be matched to $q$. For any point $q$, let $C(q,j)$ be the cost of an OMMD for the set of points $\{p\in S\cup T$, with $p\leq q$ and $1 \leq j \leq Demand\left(q\right)\}$. Initially $C\left(q,j\right)=\infty $ for all $q\in S\cup T$ and $1\le j\le Dem(q)$. If $m$ and $m'$ are the largest point and largest demand in $S\cup T$, respectively then $C(m,m')$ is the cost of an OMMD.

Our idea is that in an OMMD there exists an intersection, since the points are previously used for satisfying the demands of the intersecting points, and they are unusable for the previously matched points. In fact, the $kth$ best point of a point is the $k'th$ best point of another point, and so we have an intersection. Another reason of an intersection is that a point for satisfying its demand is matched to points that are far and so, other points can satisfy their demands from closer points as stated in Observation \ref{observation1}, Observation \ref{observation2}, and observation \ref{observation3}. Consider the first demands of the points of $A_{w+1}$. Some points of $A_w$ are matched with $b_1$, and some not, so we have two subsets in $A_w$: $A_{w11}=a_1 \dots a_j$ that is not matched with the points in $A_{w+1}$, and $A_{w21}=a_{j+1} \dots a_s$ that is matched to the points in $A_{w+1}$. Then, consider the second demands of the points of $A_{w+1}$. The points of $A_{w+1}$ that effect the first partition, $A_{w11}$, are different from the points in $A_{w+1}$ that effect $A_{w21}$. So, we must find the optimal points of $A_{w11}$ and $A_{w21}$ separately. Consider the $k$th demands of the points of $A_{w+1}$, we have $k+1$ partitions $A_{w11},A_{w2},\dots,A_{w(k+1)1}$ after inserting the point $b_1$. We proceed until satisfying the demands of all points in $A_{w+1}$ and finding an OMMD for all points $p \le b_i$. Obviously, in this algorithm there exists two steps: main and final.

\begin{lemma}
\label{lem4_22}
Each point $b_i \in A_{w+1}$ partitions $A_w$ to $k+1$ partitions, $A_{wii}, \dots, A_{w(k+1)i}$, where $k$ is the largest demand of the points of $A_w$ and $m'$ is the largest demand of the points of $S \cup T$. Each sub partition $A_{wji}$ contains a point $q(w,j,i)$ for $i\leq j \le k$, such that all points $a \in A_{wji}$ with $a<q(w,j,i)$ and $demand(a) \ge j$ are matched with the points in $A_{w-1}$ and all points $a' \in A_{wji}$ with $q(w,j,i)<a'$ and $demand(a') \ge j$ are matched with the points in $A_{w+1}$.
\end{lemma}

\begin{proof}
Given $b\le a <a'\le b'$, when there exists an intersection, that is the pairs $(b,a')$ and $(a',b)$ are in $M$, we have $(b,a)\in M$ or $(a',b') \in M$ or both. Consider the first demands of the points in $A_w$. Since the points of $A_w$ are not matched with any point previously, the points that can be matched with them are the same. For the second demands the points of $A_w$, there exists two different sub partition in $A_w$, the partition that is matched to the points in $A_{w+1}$ for their first demand and the sub partition that is matched to $A_{w-1}$ for their first demand. So, theses two sub partitions can be matched to different points for their second demand and so their separating points should be computed independently. We can prove it for the $k$th demands of points as the same. So, we can find the separating points of each subset independent from other subsets.

\end{proof}

By Corollary \ref{corlem5}, for each matching $(a,d)$ with $a \in A_i$ and $d \in A_j$ in an OMMD we have $j=i+1$ except two cases: (i) $d \in A_j$ is matched with all points of $A_{j-1}, A_{j-3}, \dots, A_{i+2}$, and (ii) $a \in A_i$ is matched with all points of $A_{i+1}, A_{i+3}, \dots, A_{j-2}$. In the first case, considered in Case B of the main step of our algorithm, we seek the partitions $A_{j-1},A_{j-3}, A_{j-5}, \dots$ to satisfy the demands of $d$. In the second case, explained in the final step of our algorithm, we seek the partitions $A_{i+1},A_{i+3},A_{i+5},\dots$ to satisfy the demands of the point $a$. In the following, we explain the main step our algorithm.
\newline
\newline

\newtheorem{observation}{Observation}
\begin{observation}
\label{observation1}
Let $b\le a <a'\le b'$, in a minimum cost OMMD $M$, if ${(b,a),(a',b),(a,b')}\subseteq M$, then: (i) $(b,a)$ and $(b',a)$ are used for satisfying the demands of $a$, (ii) either $b$ is closer to $a'$ or $a'$ satisfies the demand of $b$ (Figure \ref{fig:observ1}).
\end{observation}

\begin{observation}
\label{observation2}
Let $b\le a <a'\le b'$, in a minimum cost OMMD $M$, if ${(a',b'),(a',b),(a,b')}\subseteq M$, then: (i) $(b,a')$ and $(b',a')$ are used for satisfying the demands of $a'$, (ii) either $b'$ is closer to $a$ or $a$ satisfies the demand of $b'$ (Figure \ref{fig:observ2}).
\end{observation}

\begin{observation}
\label{observation3}
Let $b\le a <a'\le b'$, in a minimum cost OMMD $M$, if ${(b,a),(a',b),(a,b'),(a',b')}\subseteq M$, then: (i) these pairs are used for satisfying the demands of at least one of the sets $\{a,b\}$, $\{a',b'\}$, $\{a,a'\}$ or $\{b,b'\}$ (Figure \ref{observ3}).
\end{observation}

\begin{figure}
\vspace{-3cm}
\hspace{-6cm}
\resizebox{2\textwidth}{!}{%
  \includegraphics{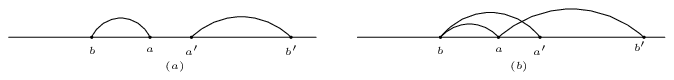}
}
\vspace{-27cm}
\caption{Illustration of Observation \ref{observation1}.}
\label{fig:observ1}       
\end{figure}

\begin{figure}
\vspace{-3cm}
\hspace{-6cm}
\resizebox{2\textwidth}{!}{%
  \includegraphics{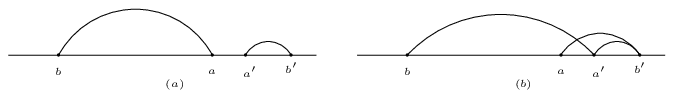}
}
\vspace{-27cm}
\caption{Illustration of Observation \ref{observation2}.}
\label{fig:observ2}       
\end{figure}

\begin{figure}
\vspace{-3cm}
\hspace{-6cm}
\resizebox{2\textwidth}{!}{%
  \includegraphics{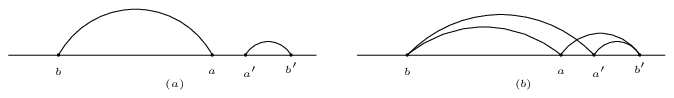}
}
\vspace{-27cm}
\caption{Illustration of Observation \ref{observation3}.}
\label{observ3}       
\end{figure}

\textbf{Main step.} Our idea is that the points $a \in T \cup S$ are examined one by one. If a point decreases the cost of matching it is selected by other points, automatically. So, we first examine that if that point is selected by other points. Then if the demand of that point is satisfied we are done, otherwise we must satisfy the demands of it such that the cost is minimized. So, we first examine that if the point is selected by other points for their first demand, then we examine that if the point is selected by other points for their second demand, and so on.

Consider $A_{w-1}=\{b'_1,b'_2,\dots,b'_r\}$, $A_w=\{a_1,a_2,\dots,a_s\}$ and $A_{w+1}=\{b_1,b_2,\dots,b_t\}$. Let $D_w=\{{\alpha}_1,{\alpha}_2,\dots,{\alpha}_s\}$ and $D_{w-1}=\{{\beta}_1,{\beta}_2,\dots,{\beta}_t\}$ be the demand sets of the points in $A_w$ and $A_{w-1}$, respectively. Assume that we have computed the separating points $q(w-1,j,i)$ for the set $A_{w-1}$ using two partitions $A_{w-2}$ and $A_{w}$ such that the cost is minimized without considering the demands of $A_w$, and now we want to compute $q(w,j,k)$ for $A_w$ using the points in $A_{w-1}$ and $A_{w+1}$.

In this step, the main step, we have two general cases: Case A, and Case B. In the first case, Case A, we assume that the demands of $A_w$ can be satisfied by the points in $A_{w-1}$ or $A_{w+1}$, so we have $\max(r,t)\geq \max^s_{j=1}{{\alpha}_j}$. So for each matching $(a,d)$ we have $a \in A_i$ and $d \in A_{i+1}$. But in the second case, Case B, we have $\max^s_{j=1}{{\alpha}_j}< \max(r,t)$. Therefore, by Corollary \ref{corlem5}, we must investigate the partitions $A_{w-2}, A_{w-4},A_{w-6}, \dots$ to find the points for satisfying the $k+1$th demands of the points.

\begin{description}
\item[Case A:] $\max^s_{j=1}{{\alpha}_j}\leq  \max (r,t)$. That is, all demands of $A_w$ can be satisfied by $A_{w-1}$ or $A_{w+1}$.
\begin{description}

\item[Case A.0:] $w=0$. In this case, $A_0$ only can be matched to $A_1$. So, we get the minimum matching by satisfying the demands of each point in $A_0$ to the smallest unmatched point in $A_1$ without considering the demands of $A_1$.

\item[Case A.1:] $w>0$. In this case, we find the separating points of $A_w$ as follows. Consider the first point of $A_{w+1}$, that is $b_1$, it might decrease the cost of satisfying the first demands of the points of $A_w$. So, we find the separating point for the first demand of the points $a_1, \dots, a_s$. Let $a_h$ be the separating point. That is by matching the points $a_h, \dots,a_s$ to $b_1$ the cost of OMMD decreases. Then, we consider the points $a_1, \dots, a_{h-1}$, and examine that if $b_1$ is selected by them for their second demands. We repeat this process for other demands and examine that if $b_1$ can decrease the cost of matching or not. Then, we insert $b_2$. Let $a_i, \dots, a_j$ be the points that are matched to $b_1$ for their $k$th demands. We examine that if $b_2$ can decrease the cost of matching by satisfying their $k+1$th demands or not. We examine all the points $b_1, \dots, b_t$ in $A_{w+1}$ to find the separating points $q(w,j,i)$ of $A_w$ for $1 \leq k\leq m'$ and $i \leq j\leq k$.

    \begin{lemma}
\label{lemma-new2}
Give $A_w=\{a_1,a_2,\dots,a_s\}$ with $a_1< a_2<\dots<a_s$ and $A_{w-1}=\{b'_1,b'_2,\dots,b'_r\}$ with $b'_1< b'_2<\dots<b'_r$. In an OMMD if $p_i$ denotes the $i$th point in $A_w$ with $demand(p_i)\ge k+1$ that are affected by the same subsets of $A_{w-1}$. Assume that we have computed the separating points of $A_w$ for the first to $k$th demands, and now we want to compute the separating point of $A_w$ for the $k+1$th demands. Assume that we want to examine that if $p_h$ is the separating point or not. So, we must match $p_1, \dots, p_{h-1}$ with the points in $A_{w-1}$ as follows. There exist three cases:
\begin{itemize}
\item There exists $a \in A_w$ with $deg(a)>demand(a)$. Remove an arbitrary pair $(a,bb)$ form matching and add $(p_i,bb)$ to the matching.

 \item $p_{i-1}$ is matched to a point $b'' \in A_{w-1}$ with $deg(b'')>demand(b'')$. Then $p_i$ is matched to the largest unmatched point in $A_{w-1}$.
 \item $p_{i-1}$ is matched to a point $b'' \in A_{w-1}$ with $deg(b'')=demand(b'')$. $p_i$ tests whether be matched to the largest unmatched point $b'_q$ in $A_{w-1}$ with $deg(b'_q)>demand(b'_q) $ or the largest unmatched point $b'_l$ in $A_{w-1}$ with $deg(b'_l)=demand(b'_l)$.

\end{itemize}

\end{lemma}

\begin{proof}

This Lemma is proved similar to Lemma \ref{lemma-new}.
\end{proof}

We examine $p_1, p_2, p_3, \dots$ using Lemma \ref{lemma-new2} to find the separating point.

\end{description}

\item[Case B:] $w>0$, $\max^i_{j=1}{{\beta }_j}>s$. In this case, the number of the points in $A_w$ is less than the maximum demand number of the points in $A_{w+1}$, so by Corollary \ref{corlem5} we should seek the previous sets to satisfy demands of the points $b_1,b_2,\dots,b_i$.

This backward process is followed until finding the first set, called $A_{w'}$, that can satisfy the demands of the points in the set $A_{w+1}$. Then, we must find OMMD for the sets $A_{w-1} \dots A_{w'}$ again. It is possible that we do not reach such a set $A_{w'}$, in this situation $C\left(b_i,k\right)=\infty $ for all $1\le k\le {\beta}_i$.

\end{description}

\textbf{Final step.} In this step, we consider the situation where there are not enough points in $A_{w+1}$ for demands of the points $a\in A_w$ for $w \geq 0$, and so by Corollary \ref{corlem5}, we must seek the partitions $A_{w+3},A_{w+5},\dots$ for finding new points.

This forward process is followed until finding the first partition, called $A_{j'}$, which can satisfy the demands of the points in $A_w$. Then, we must find OMMD for the sets $A_{w+2} \dots A_{j'}$ again. It is possible that there does not exist a set such $A_{j'}$. If $A_{j'}$ exists, we should match the unmatched points in $A''$ with the points in $A_{j'}$ as follows.

\begin{lemma}
\label{lem8}
Given $A_w={a_1,a_2, \dots, a_s}$, let $C(a_i,k)$ be the cost of satisfying $k$ number of demands of the first $i$ points in $A_w$ with the points in $A_{j}$ with $j<i$. That is, assume that $C(a_i,k)$ denotes the cost of matching $k$ demands of $a_1, a_2,\dots,a_i$ with the points $p$ where $p<=a_1$. Also let $Cost(a_i,k-j)$ be the cost of matching $k-j$ demands of $a_i$ with the points $p$ where $p<=a_1$.
Then we have:
$$C(a_i,k)=\min^k_{j=1}(C(a_{i-1},j)+Cost(a_i,k-j)).$$

\end{lemma}

\begin{proof}

Given $A_i={a_1,a_2, \dots, a_q}$, for $a_i$ there exists $k$ cases, depending on satisfying its demands with $p<a_i $ or $p'>a_i$. Assume that $j$ demands of $a_i$ is satisfied with the points $p<a_i $. Then, by Lemma \ref{lem6} the points $a_1, a_2, \dots, a_{i-1}$ have only $k-j$ options for satisfying their demands (see Figure \ref{fig:333}), since $j$ demands of them must be satisfied with the $j$ points $a_i$ is satisfied with them.

\begin{figure}
\vspace{-3cm}
\hspace{-6cm}
\resizebox{2\textwidth}{!}{%
  \includegraphics{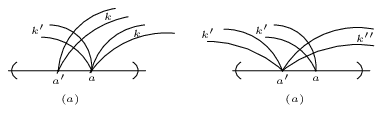}
}
\vspace{-27cm}
\caption{Intersection.}
\label{fig:333}       
\end{figure}

\end{proof}
\begin{lemma}
\label{lem-tn}
Given the partition $A_w={a_1,a_2, \dots, a_s}$ of points with maximum demand $k$, let $T(s,k)$ be the number of options for satisfying $k$ demands of $a_1,a_2, \dots, a_i$, then:
$$T(s,k)=\sum^k_{j=1}{T(s-1,j)}$$
\end{lemma}
\begin{proof}
For demands of $a_s$, it can satisfy $j$ number of its demands by previous partitions. As Lemma \ref{lem8} this lemma is proved.

\end{proof}

So, by Lemma \ref{lem8} using dynamic programming we have $n*k$ states, that one of them is \textit{the optimal state}, denoted by $optimal(b_t,k)$. That is, the state that minimizes the cost of matching.

We can compute an OMMD by finding the optimal state $optimal(b_t,k)$ of each set $A_i$.

 \section{Concluding Remarks}
 \label{Conclsection}

Many-to-many point matching with demands MMD is a many-to-many matching where each point has a demand. We studied the one dimensional MMD, called OMMD, where we match two point sets on the line. We presented an algorithm for getting an OMMD between two point sets with total cardinality $n$ in $O(n^2)$ time.



\begin{thebibliography}{}
%
%



\bibitem{Ben-Dor}
Ben-Dor, A., Karp, R.M., Schwikowski, B., Shamir, R.: The restriction scaffold problem. J. Comput. Biol. \textbf{10,} 385-398 (2003)


\bibitem{Burkard}
Burkard, R.E., C¸ela, E.: Linear assignment problems and extensions. in: D.-Z. Du, P.M. Pardalos (Editors), Handbook of Combinatorial Optimization (Supplement Volume A) (volume 4). Kluwer Academic Publishers, Dordrecht, 75-149 (1999)

\bibitem{Buss}
Buss, S.R., Yianilos, P.N.: A bipartite matching approach to approximate string comparison and search. , Technical report, NEC Research Institute, Princeton, New Jersey (1995)

\bibitem{ColanDamian}
Colannino, J., Damian, M., Hurtado, F., Langerman, S., Meijer, H., Ramaswami, S., Souvaine, D., Toussaint, G.: Efficient many-to-many point matching in one dimension. Graph. Combinator. \textbf{23,} 169-178 (2007)

\bibitem{ColanToussa}
Colannino, J., Toussaint, G.: An algorithm for computing the restriction scaffold assignment problem in computational biology. Inform. Process. Lett. \textbf{95(4),} 466-471 (2005)


\bibitem{ColanninoToussaFaste}
Colannino, J., Toussaint, G.: Faster algorithms for computing distances between one-dimensional point sets. in: Francisco Santos and David Orden (editors). In Proc. XI Encuentros de Geometria Computacional, 189-198 (2005)



\bibitem{Eiter}
Eiter, T., Mannila, H.: Distance measures for point sets and their computation. Acta Informatica. \textbf{34,} 109-133 (1997)

\bibitem{Fatih}
Demirci, M.F., Shokoufandeh, A., Keselman, Y., Bretzner, L., Dickinson, S.: Object recognition as many-to-many feature matching, Int. J. Comput. Vision. \textbf{69,} 203-222 (2006)


\bibitem{Karp}
Karp, R.M., Li, S.-Y.R.: Two special cases of the assignment problem. Discrete Math. \textbf{13(46),} 129-142 (1975)


\bibitem{Rajabi-Alni}

Rajabi-Alni, F., Bagheri, A.: An O(n2) Algorithm for the Limited-Capacity Many-to-Many Point Matching in One Dimension. Algorithmica 76(2): 381-400 (2016)


\bibitem{Schrijver}
Schrijver, A.: Combinatorial optimization. polyhedra and efficiency, vol. A, Algorithms and Combinatorics, no. 24. Springer-Verlag, Berlin (2003)


 \end{thebibliography}


\end{document}